\documentclass[12pt,draftcls,journal,onecolumn]{IEEEtran}
\usepackage{amsfonts,color,morefloats,pslatex,graphicx,graphics}
\usepackage{amssymb,amsthm, amsmath,latexsym}

\newtheorem{theorem}{Theorem}
\newtheorem{lemma}[theorem]{Lemma}
\newtheorem{remark}[theorem]{Remark}

\newtheorem{example}[theorem]{Example}

\newcommand{\ord}{{\mathrm{ord}}}

\newcommand{\lcm}{{\mathrm{lcm}}}

\newcommand{\gf}{{\mathrm{GF}}}

\newcommand{\F}{{\mathbb{F}}}

\newcommand{\C}{{\mathcal{C}}}

\usepackage{blindtext}

\ifCLASSINFOpdf

\else

\fi

\begin{document}
\title{The duals of narrow-sense BCH codes \\ with length $\frac{q^m-1}{\lambda}$
\thanks{
%The work of Xiaoqiang Wang was supported by the National
%Natural Science Foundation of China (12001175).
%The work of Chengju Li was supported by the National
%Natural Science Foundation of China (12071138), Shanghai Rising-Star Program (22QA1403200), the open research fund of National Mobile Communications Research Laboratory of Southeast University (2022D05), and the Shanghai Trusted Industry Internet Software Collaborative Innovation Center.
%The work of Yansheng Wu was supported by the National Natural Science Foundation of China (12101326) and the Natural Science Foundation of Jiangsu Province (BK20210575 ).
}
}
\author{Xiaoqiang Wang,\,\,\,\,\,\, Chengliang Xiao, \,\,\,\,\,\,Dabin Zheng{\thanks{Hubei Key Laboratory of Applied Mathematics, Faculty of Mathematics and Statistics, Hubei University, Wuhan 430062, China (Email: waxiqq@163.com, xichll@163.com, dzheng@hubu.edu.cn). The corresponding author is Dabin Zheng.
}}
%\thanks{\it(Corresponding Author: Dabin Zheng.)}
}

%\author{Xiaoqiang Wang, Jiaojiao Wang, Chengju Li, Yansheng Wu{\thanks{
%Xiaoqiang Wang is with Hubei Key Laboratory of Applied Mathematics, Faculty of Mathematics and Statistics, Hubei University, Wuhan 430062, China (E-mail: waxiqq@163.com).}}}
%
%\thanks{Jiaojiao Wang is with Data Science and Information Technology Research Center, Tsinghua-Berkeley Shenzhen Institute, Tsinghua Shenzhen International Graduate School, Shenzhen, China. (E-mail:  wjj22@mails.tsinghua.edu.cn).}
%
%\thanks{ Chengju Li is the Shanghai Key Laboratory of Trustworthy Computing, East China Normal University,
%Shanghai, 200062, China; and is also with the National Mobile Communications Research Laboratory, Southeast University, Nanjing 210096, China  (E-mail:  cjli@sei.ecnu.edu.cn).}
%
%\thanks{Yansheng Wu is with School of Computer Science, Nanjing University of Posts and Telecommunications, Nanjing 210023, China (E-mail: yanshengwu@njupt.edu.cn).}
%}
%
%}

\maketitle

\begin{abstract} BCH codes are an interesting class of cyclic codes due to their efficient encoding and decoding algorithms. In the past sixty years,
a lot of progress on the study
of BCH codes has been made, but little is known about the properties of their duals.  Recently, in order to study the duals of BCH codes and the lower bounds on their minimum distances, a new concept called dually-BCH code was proposed by authors in \cite{GDL21}.
In this paper, the lower bounds on the minimum distances of the duals of narrow-sense BCH codes with length  $\frac{q^m-1}{\lambda}$ over $\mathbb{F}_q$ are developed, where $\lambda$ is a positive integer satisfying $\lambda\, |\, q-1$, or $\lambda=q^s-1$ and $s\, |\,m$.
In addition, the sufficient and necessary conditions in terms of the designed distances for these codes being dually-BCH codes are presented. Many considered codes  in \cite{GDL21} and \cite{Wang23} are the special cases of the codes showed in this paper.
 Our lower bounds on the minimum distances of the duals of BCH codes include the bounds stated in \cite{GDL21} as a special case. Several examples show that the lower bounds are good in some cases.
\end{abstract}

%The question as to whether the dual code of a BCH code is still a BCH code very hard to answer in general. We call a BCH code is a dually-BCH code if its dual is also a BCH code.


% \fi
\IEEEpeerreviewmaketitle

\section{Introduction}\label{sec-auxiliary}

Let $q$ be a prime power and $\mathbb{F}_q$ be the finite field with $q$ elements.
Let $n$, $k$ be two positive integers such that $k\leq n$.  An $[n, k, d]$ code
$\mathcal{C}$ over the finite field $\mathbb{F}_q$ is a $k$-dimensional
linear subspace of $\mathbb{F}_q^n$ with minimum distance $d$. Let
$$\mathcal{C}^{\perp}=\{\mathbf{b} \in \mathbb{F}_q^n\,:\,\mathbf{b}\mathbf{c}^{T}=0 \,\,\text{for any $\mathbf{c} \in \mathcal{C}$}\}, $$
where $\mathbf{b}\mathbf{c}^{T}$ is the standard inner product of two vectors $\mathbf{b}$ and $\mathbf{c}$.
Then $\mathcal{C}^{\perp}$ is called {\it the dual} of $\mathcal{C}$. The code $\mathcal{C}$ is said to be {\it cyclic} if
$(c_0,c_1,  \ldots, c_{n-1}) \in \C$ implies $( c_{n-1}, c_0, c_1, \ldots, c_{n-2})
\in \C$. By identifying any vector $(c_0,c_1,\ldots,c_{n-1}) \in \mathbb{F}_q^n$ with
$$c_0+c_1x+c_2x^2+\ldots+c_{n-1}x^{n-1} \in \mathbb{F}_q[x]/(x^n-1),$$
the cyclic code $\C$ of length $n$ over $\mathbb{F}_q$ corresponds to an ideal of the quotient ring
$\mathbb{F}_q[x]/(x^n-1)$.
Since the ideal of $\mathbb{F}_q[x]/(x^n-1)$ is principal, $\C$ can be expressed as $\C=\langle g(x)\rangle$, where
$g(x)$ is a monic polynomial with the smallest degree and is called the {\it generator polynomial} of $\C$.

In this paper, we always assume that $\gcd(n,q)=1$. Let $\ell=\ord_{n}(q)$ be the order of $q$ modulo $n$
and $\alpha$ be a generator of the group $\mathbb{F}_{q^\ell}^*$. For any $i$ with $0\leq i\leq q^\ell-2$,
let $m_i(x)$ denote the {\it minimal polynomial} of $\beta^i$ over $\mathbb{F}_{q}$, where $\beta=\alpha^{(q^\ell-1)/n}$ is a primitive $n$-th root of unity. Let $\C_{(\delta,b)}$ be the cyclic code with generator polynomial
\begin{eqnarray*}
g_{(\delta,b)}(x)=\lcm \left(m_{b}(x), m_{b+1}(x), \cdots, m_{b+\delta-2}(x)\right),
\end{eqnarray*}
where $2\leq \delta\leq n$, $b$ is an integer and lcm denotes the least common multiple of these minimal polynomials. Then $\C_{(\delta,b)}$ is called a {\it BCH code} over $\F_q$ with length $n$ and {\it designed distance} $\delta$.
If $b=1$, the code $\C_{(\delta,1)}$ is called a {\it narrow-sense BCH code} and for convenience we abbreviate it to $\C_{\delta}$ in the sequel.  If the length of $\C_{(\delta,b)}$ is $q^m-1$, or $q^m+1$, or $\frac{q^m-1}{q-1}$, then $\C_{(\delta,b)}$ is called a {\it primitive BCH code}, or {\it antiprimitive BCH code}, or {\it projective BCH code}, respectively.

BCH codes are an error correction coding technique commonly used in data transmission and storage to improve reliability. Binary BCH codes were discovered by Bose, Ray-Chaudhuri, and  Hocquenghem around 1960 in \cite{Bose62,Hocquenghem59}, and were generalized to the case of all finite fields by Gorenstein and Zierler in 1961 \cite{Gorenstein61}. In the past several decades, BCH codes have been extensively studied. However, the parameters of BCH codes are known only for a few special cases. Among all types of BCH codes, narrow-sense primitive BCH codes are the most extensively
studied. The reader is referred to, for example, \cite{Aly07,Augot94,Charpin90,Ding15,Ding17,Lid17,Liu17, Yue15,Dianwu96} for information.
Antiprimitive BCH codes are another family of interesting codes. The reader is referred to \cite{Lid017,Li2019,Liu17,Yan2018}. For further information on BCH codes of the other lengths, the reader is referred to \cite{Li2017, Yan2018, Ling23,Zhu19}.

Until now, we have very limited knowledge about the minimum distances of the duals of BCH codes. In \cite{MacWilliams77}, the authors showed two lower bounds on the minimum distances of the dual codes of binary primitive BCH codes with odd designed distance. In \cite{Augot96}, the authors presented the lower bounds on the minimum distances of the duals of primitive BCH codes via the adaptation of the Weil bound to cyclic codes. In \cite{GDL21} and \cite{Wang23}, the authors obtained the lower bounds on the minimum distances of the duals of $\C_{\delta}$ for length $q^m-1$, $\frac{q^m-1}{q-1}$ and $\frac{q^m-1}{q+1}$ \ ($m$ even), respectively.

In general, the dual of a BCH code is not a BCH code. However, in some specific cases, the dual of a BCH code still is a BCH code. In order to further study the properties of the duals of BCH codes, the authors in \cite{GDL21} proposed a new concept, which is called {\it dually-BCH code}: A BCH code is called a dually-BCH code if both the BCH code and its dual are BCH codes with respect to an $n$-th primitive root of unity $\beta$.
As far as we know,  all the results about the dually-BCH codes are focused on narrow-sense BCH codes until now. In \cite{GDL21}, the authors obtained a sufficient and necessary condition for a narrow-sense BCH code $\C_{\delta}$  being a dually-BCH code, where the length of the considered code is $q^m-1$ or $\frac{q^m-1}{q-1}$ ($q=3$), and showed the condition for $\C_{\delta}$  being a dually-BCH code with length $\frac{q^m-1}{q-1}$ ($q\geq 4$) as an open problem. And then \cite{Wang23} solved this open problem. In addition,
the authors in \cite{Fan23} derived a sufficient and necessary condition for above BCH codes
being Hermitian dually-BCH codes.

Let $q$ be a prime power and $\lambda$ be a positive integer satisfying $\lambda\, |\, q-1$, or $\lambda=q^s-1$ and $s \, | \, m$. The main objective of this paper is to give several sufficient and necessary conditions in terms of the designed distances to ensure that  BCH codes with length $\frac{q^m-1}{\lambda}$ are dually-BCH codes, and develop the lower bounds on the minimum distances of the duals for these BCH codes. Our theorems generalize many results in \cite{GDL21} and \cite{Wang23} since the codes with lengths $q^m-1$ and $\frac{q^m-1}{q-1}$ are the special cases of our codes. To investigate the lower bounds on the minimum distances of the duals of the codes studied in this paper, we determined the true minimum distances of these dual codes for some special lengths by Magma, and we found that our lower bounds are tight for some cases.

%  In addition, our bounds on the minimum distances of the duals improve some well known bounds for some special designed distances.

The rest of this paper is organized as follows. Section II contains some preliminaries. Section III gives the sufficient and necessary conditions in terms of the designed distances to ensure that the BCH codes with length $\frac{q^m-1}{\lambda}$ are dually-BCH codes and develop the lower bounds on the minimum distances of the dual codes, where $\lambda$ is a positive integer satisfying $\lambda\, |\, q-1$, or $\lambda=q^s-1$ and $s \, | \, m$. Section IV concludes the paper.

\section{Preliminaries}\label{sec-auxiliary}

In this section, we present some
basic concepts and results, which will be empolyed later.

\subsection{Some notation and basic results on BCH codes}

Starting from now on, we adopt the following notation unless otherwise stated:
\begin{itemize}
\item $\mathbb{F}_q$ is the finite field with $q$ elements, where $q$ is a prime power.
\item $\alpha$ is a primitive element of $\mathbb{F}_{q^m}$ and $\beta=\alpha^{\frac{q^m-1}{n}}$ is a primitive $n$-th root of unity, where $n \, |\, q^m-1$.
\item $m_i(x)$ denotes the minimal polynomial of $\beta^i$ over $\mathbb{F}_q$.
\item $g_{(\delta,b)}(x)=\lcm \left(m_{b}(x), m_{b+1}(x), \ldots, m_{\delta+b-2}(x)\right)$ denotes the least common multiple of these minimal polynomials.
\item $\mathcal{C}_{\delta}$ denotes the BCH code with generator polynomial $g_{(\delta,1)}$ and length $n$.
\item $T=\{0\leq i\leq n-1: g_{(\delta,1)}(\beta^i)=0\}$ is the defining set of $\mathcal{C}_{\delta}$ with respect to $\beta$.
\item $T^{-1}=\{n-i\,:\,i\in T\}$.
\item $T^{\perp}$ is the defining set of the dual code $\mathcal{C}^{\perp}_{\delta}$ with respect to $\beta$.
%\item  $m$ and $s$ are positive integers such that $s\, |\,m$.
\item ${\rm CL}(a)$ denotes the $q$-cyclotomic coset leader modulo $n$ containing $a$, where $a$ is a positive integer with $1\leq a< n$.
\item  $\mathbb{Z}_n$ denotes the ring of integers modulo $n$.
%\item $\lceil x \rceil$ denotes the smallest integer larger than or equal to $x$.
%\item $\lfloor x \rfloor$ denotes the largest integer less than or equal to $x$.
%\item $\overline{x}_t$ denotes $x \pmod{t}$ and $0\leq\overline{x}_t\leq t-1$, where $t$ is a positive integer.
\end{itemize}

Let $s$ be an integer with $0\leq s<n$. The {\it $q$-cyclotomic coset} of $s$ modulo $n$ is defined by
$$ \mathbb{C}_{s}=\{s,sq,sq^2,\ldots,sq^{\ell_{s-1}}\}\,\, {\text\,\,mod \,\,n\subseteq \mathbb{Z}_n, }$$
where $\ell_s$ is the smallest positive integer such that $s\equiv sq^{\ell_s} \pmod n$, and is the size of the $q$-cyclotomic coset.
The smallest integer in $\mathbb{C}_{s}$ is called the {\it coset leader} of $\mathbb{C}_{s}$.

For a positive integer $i$ with $0<i<q^m-1$, the $q$-adic expansion of $i$ can be written as
 $$i=(i_{m-1},i_{m-2},\ldots,i_{0})_q.$$ Let $0\leq j\leq m-1$.
It is easily seen that the $q$-adic expansion of $iq^j\pmod {q^m-1}$ is
$$iq^j\,\,\, (\text{mod} \,\,\,q^m-1)= (i_{m-j-1},i_{m-j-2},\ldots, i_{m-j})_q,$$
which is called the {\it circular $j$-left-shift} of $(i_{m-1},i_{m-2},\ldots,i_0)_q$, where the subscript of each coordinate is regarded as an integer modulo $m$. Let $0<a,b\leq q^m-1$ be two positive integers with $q$-adic expansion $$a=(a_{m-1},a_{m-2},\ldots,a_0)_q\,\,\,\text{and}\,\,\,b=(b_{m-1},b_{m-2},\ldots,b_0)_q.$$
 We say that  \begin{equation}\label{eq1120}
(a_{m-1},a_{m-2},\ldots,a_0)_q\geq (b_{m-1},b_{m-2},\ldots,b_0)_q
\end{equation} if and only if there exists an integer $0\leq l \leq m-1$ such
that $a_l\geq b_l$ and $a_j=b_j$ for $ l+1\leq j \leq m-1$. Then  $a\geq b$ if and only if (\ref{eq1120}) holds.
With the preparations above, we have the following result, which can be derived directly.

\begin{lemma}\label{lem1b21}
Let $a,b$ be defined as above. Then the coset leader of  $\mathbb{C}_a$ modulo $q^m-1$ is greater than or equal to $b$ if and only if the circular $j$-left-shift of $(a_{m-1},a_{m-2},\ldots,a_0)$ is greater than or equal to
$(b_{m-1},b_{m-2},\ldots,b_0)$ for each $0 \leq j \leq m-1$.
\end{lemma}

In the following, we will introduce some basic results about $q$-cyclotomic cosets. By the same way as \cite[Lemma 6]{Wang19}, we have the following results.

\begin{lemma}\label{lem:0913}
Let $0<t<q^m-1$ be a positive integer. Let $\mu$ be a common factor of $t$ and $q^m-1$, then $t$ is a coset leader modulo $q^m-1$ if and only if $\frac{t}{\mu}$ is a coset leader modulo $\frac{q^m-1}{\mu}$.
\end{lemma}

\begin{lemma} \label{lem:0914}
Let $a$ and $b$ be positive integers in the same $q$-cyclotomic coset modulo $q^m-1$. If $\lambda \, | \, a$, $\lambda\, | \, b$ and $\lambda \, | \, q^m-1$, then $\frac{a}{\lambda}$ and $\frac{b}{\lambda}$ are in the same $q$-cyclotomic coset modulo $\frac{q^m-1}{\lambda}$.
\end{lemma}

\begin{proof}
Since $a$ and $b$ in the same $q$-cyclotomic coset modulo $q^m-1$, there exists a positive integer $i$ such that
$$aq^i\equiv b \pmod {q^m-1}.$$
Then,
$$\frac{aq^i}{\lambda}\equiv \frac{b}{\lambda} \pmod {\frac{q^m-1}{\lambda}}.$$
Hence, $\frac{a}{\lambda}$ and $\frac{b}{\lambda}$ in the same $q$-cyclotomic coset modulo $\frac{q^m-1}{\lambda}$.
The proof is then completed.
\end{proof}

\begin{lemma}\label{lem:qm1q1}\cite{Ding17}
Let $q$ be a prime power. The first three largest $q$-cyclotomic coset leaders modulo $q^m-1$ are as follows:
$$\delta_1=(q-1)q^{m-1}-1,\,\,\delta_2= (q-1)q^{m-1}-q^{\lfloor\frac{m-1}{2}\rfloor}-1,\,\,\delta_3= (q-1)q^{m-1}-q^{\lfloor\frac{m+1}{2}\rfloor}-1.$$
\end{lemma}

\begin{lemma}\label{lem:qm1q101}\cite{Wang23}
Let $q\geq 3$ be a prime power and $m\geq 4$ be an integer. Then,
$\delta_1=q^{m-1}-1-\frac{(\sum_{t=1}^{q-2}q^{\lceil\frac{mt}{q-1}-1\rceil}-q+2)}{q-1}$
is the first largest $q$-cyclotomic coset leader modulo $\frac{q^m-1}{q-1}$.
\end{lemma}

From Lemma \ref{lem:0913} and Lemma \ref{lem:qm1q1}, the following result can be derived directly.

\begin{lemma}\label{lem:qm1q102}
Let $\delta_1$ and $\delta_2$ be the first two largest $q$-cyclotomic coset leaders modulo $\frac{q^m-1}{2}$, then
$$\delta_1= \frac {(q-1)q^{m-1}-q^{\lfloor\frac{m-1}{2}\rfloor}-1}{2},\,\,\delta_2= \frac {(q-1)q^{m-1}-q^{\lfloor\frac{m+1}{2}\rfloor}-1}{2}.$$
\end{lemma}

%\begin{lemma}
%If $gcd(n,q)=1$ and $q^{\lfloor\frac{m}{2}\rfloor} < n\leq q^m-1$, then $s\in MinRepn$ and $|C_s|=m$ for any $1\leq s\leq \frac{nq^{\lfloor\frac{m}{2}\rfloor}}{q^m-1}$, $q\nmid s$.
%
%
%\end{lemma}

%\begin{lemma}\label{lem1b21}
%Let $a$ be given as above and $c$ be the coset leader of $C_a$. Let $0<b<q^m-1$ and $b=(b_{m-1},b_{m-2},\cdots,b_{0})_q$.  Then $c<b$ if and only if there exists an integer $0\leq i\leq m-2$ such that $c_i> b_i$ and $c_j=b_j$ for $j\in [i+1,m-1]$. $c_{m-1}>b_{m-1}$, or there exists an integer $0\leq i\leq m-2$ such that $c_i> b_i$ and $c_j=b_j$ for $j\in [i+1,m-1]$.
%\end{lemma}

%It is clear that
%for a positive integer $a$ with $0<a<q^m-1$, the $q$-adic expansion of $a$ can be written as
% $a=(a_{m-1},a_{m-2},\cdots,a_{0})_q$. Let $b=(b_{m-1},b_{m-2},\cdots,b_{0})_q$, then we have the following result directly.
%
%
%
%\begin{lemma}\label{lem1b21}
%Let $a,b$ be given as above. Then $a>b$ if and only if $a_{m-1}>b_{m-1}$, or there exists an integer $0\leq i\leq m-2$ such that $a_i> b_i$ and $a_j=b_j$ for $j\in [i+1,m-1]$.
%\end{lemma}
% Let $aq^j\equiv A\pmod{q^m-1}$. Then
%$$A=(a_{m-j-1},a_{m-j-2},\cdots,a_{0},a_{m-1},\cdots, a_{m-j})_q\,\,\, \text{or}\,\,\,A=\sum_{i=1}^{m-j}a_{m-j-i}q^{m-j}+\sum_{i=0}^{j-1}a_{m-j+i}q^{i}, $$
%where $1\leq j\leq m-1$.

% Let $aq^j\equiv A\pmod{q^m-1}$. Then
%$$A=(a_{m-j-1},a_{m-j-2},\cdots,a_{0},a_{m-1},\cdots, a_{m-j})_q\,\,\, \text{or}\,\,\,A=\sum_{i=1}^{m-j}a_{m-j-i}q^{m-j}+\sum_{i=0}^{j-1}a_{m-j+i}q^{i}, $$
%where $1\leq j\leq m-1$.

\subsection{Some konwn bounds on the minimum distance of the duals of BCH codes}
 Charpin in \cite{Charpin98} pointed out that it is a well-known hard problem to determine the minimum distances of BCH codes in general. In fact, it is also a very difficult problem to determine the minimum distances of the duals of BCH codes. The following are two known lower bounds on the minimum distance of cyclic codes~\cite{MacWilliams77}.

\begin{lemma}(Carlitz-Uchiyama Bound)
Let $\mathcal{C}$ be a binary BCH code of length $2^m-1$ with designed distance $2s+1$. Then, the minimum distance of its dual satisfies
$$d^{\perp}\geq 2^{m-1}-(s-1)2^{\frac{m}{2}}.$$
\end{lemma}

\begin{lemma}(Sidel'nikov Bound)
Let $\mathcal{C}$ be a binary BCH code of length $2^m-1$ with designed distance $2s+1$. Then, the minimum distance of its dual satisfies
$$d^{\perp}\geq 2^{{m-1}-\lfloor log_2(2s-1)\rfloor}.$$
\end{lemma}

In last year, the authors in \cite{GDL21} investigated the lower bounds on the minimum distances of the duals of $\C_{\delta}$ for length $q^m-1$ and $\frac{q^m-1}{q-1}$.
From \cite[Lemmas 15, 20 and Theorem 28]{GDL21} and BCH bound, we have the following results.

\begin{lemma}\label{lem:8}
Let $d^{\perp}(\delta)$ be the minimum distance of $\mathcal{C}^{\perp}_{\delta}$ with length $q^m-1$ over $\mathbb{F}_q$, where $m\geq 3$. Then,
\begin{equation*}
\begin{aligned}
d^{\perp}(\delta)\geq \begin{cases}
2^{m-t}, &{\rm if}\,\,\, 2^t\leq \delta < 2^{t+1}\,\,(2\leq t \leq m-3),\\
4, &{\rm if}\,\,\, 2^{m-2}\leq \delta< 2^{m-1}-2^{\lfloor\frac{m-1}{2}\rfloor}
               \end{cases}
\end{aligned}
\end{equation*}
if $q=2$, and
\begin{equation*}
\begin{aligned}
d^{\perp}(\delta)\geq\begin{cases}
q^{m-t}-a+1, &{\rm if}\,\, aq^t\leq \delta \leq (a+1)q^t-1\,\,(1\leq t \leq m-2,\,\,1\leq a<q-1),\\
q^{m-t}-q+2, &{\rm if}\,\, (q-1)q^t\leq \delta \leq q^{t+1}-q+1
\,\, (1\leq t \leq m-2),\\
q-a+1, &{\rm if}\,\,  aq^{m-1}\leq \delta \leq (a+1)q^{m-1}-1\,\,(1\leq a < q-2),\\
3, &{\rm if}\,\,  (q-2)q^{m-1}\leq \delta < (q-1)q^{m-1}-q^{\lfloor \frac{m-1}{2}\rfloor},\\
(b+1)q^{m-t},  &{\rm if}\,\,  \delta=q^t-b\,\,(1\leq t\leq m-1, 1\leq b\leq q-2, q^t-b\geq 3)
               \end{cases}
\end{aligned}
\end{equation*}
if $q\geq 3$.
\end{lemma}

\begin{lemma}\label{lem:7}
Let $d^{\perp}(\delta)$ be the minimum distance of $\mathcal{C}^{\perp}_{\delta}$ with length $\frac{q^m-1}{q-1}$ over $\mathbb{F}_q$, where $m\geq 3$. Then,
\begin{equation*}
\begin{aligned}
d^{\perp}(\delta)\geq\begin{cases}
\frac{q^{m-t}-1}{q-1}+1, &{\rm if}\,\, \frac{q^t-1}{q-1}< \delta \leq \frac{q^{t+1}-1}{q-1}\,\,(1\leq t \leq m-2),\\
2, &{\rm if}\,\, \frac{q^{m-1}-1}{q-1}< \delta< \frac{q^m-1}{q-1}.
               \end{cases}
\end{aligned}
\end{equation*}
\end{lemma}
For more information on the minimum distances of the duals of BCH codes, the reader is referred to \cite{Augot96,Wang23}.

\section{The duals of BCH codes with length $\frac{q^m-1}{\lambda}$ }\label{sec-plus}

In this section, we always assume that $n=\frac{q^m-1}{\lambda}$ is the length of the considered codes, where $m\geq 3$ and $\lambda$ is a positive integer satisfying $\lambda\, |\, q-1$,
or $\lambda=q^s-1$ and $s \, | \, m$. Let $\C_{\delta}$ be a narrow-sense BCH code over $\mathbb{F}_q$ with designed distance $\delta$. By definition, we know that the defining set of $\C_{\delta}$ with respect to $\beta$ is  $T=\mathbb{C}_1\cup \mathbb{C}_2 \cup \cdots \cup \mathbb{C}_{\delta-1}$, where $2\leq \delta\leq n$. Let $T^{\perp}$ be the defining set of the dual code $\mathcal{C}^{\perp}_{\delta}$ with respect to $\beta$. It is easy to check that $T^{\perp}=\mathbb{Z}_n \setminus T^{-1}$ and $0 \in \ T^{\perp}$.  In the following, our task is to give a lower bound on the minimum distance of the dual code $\mathcal{C}^{\perp}_{\delta}$ and derive a sufficient and necessary condition for $\mathcal{C}_{\delta}$ being a dually-BCH code. As will be seen later, the results for the cases $\lambda\, |\, q-1$, or $\lambda=q^s-1$ and $s \, | \, m$ are distinct. Hence, the proof should be treated separately.

\subsection{The case $\lambda=q^s-1$ and $s \, | \, m$ }

If $s=m$, the length of the code is $1$ and there is nothing to prove. If $m=2s$,  the length of the code is $q^s+1$ and we will obtain a trivial lower bound on the minimum distance of the dual code $\mathcal{C}^{\perp}_{\delta}$ by using our way. Hence, we always assume that $\frac{m}{s}\geq 3$ in the following.
We first prove the following lemma.

\begin{lemma}\label{eq:lemt}
Let $t$ and $u$ be positive integers such that $1\leq t \leq \frac{m}{s}-2$ and $1\leq u\leq  \frac{q^{m-ts}-1}{q^s-1}-1$. Then the coset leader of $\mathbb{C}_{q^{ts}-1+(q^s-1)uq^{ts}}$ modulo $q^m-1$ is greater than or equal to $q^{ts+s}-1$.
\end{lemma}

\begin{proof}
Let A be  the coset leader of $\mathbb{C}_{q^{ts}-1+(q^s-1)uq^{ts}}$ modulo $q^m-1$. Obviously, there is nothing to prove if $A= q^{ts+s}-1$. In the following, we assume that  $A\neq q^{ts+s}-1$. Note that \begin{equation}\label{eq:10010101}
q^{ts+s}-1 = (\underbrace{0,0,\ldots,0,}_{m-ts-s}\underbrace{q-1,q-1,\ldots,q-1}_{ts+s})_q.
\end{equation}
From Lemma \ref{lem1b21},
 there exist at least $m-ts-s$ consecutive zeros in the $q$-adic expansion of $A$ if $A$ is less than $q^{ts+s}-1$.
Hence, in order to obtain the desired result, we only need to prove that there do not exist $m-ts-s$ consecutive zeros in the $q$-adic expansion of $A$.

% In the following, for the convenience of narration, we only proof the result is hold for the case that For the other cases, the proof is similar and we omit the details.

It is easy to check that the $q$-adic expansion of $q^{ts}-1+(q^s-1)uq^{ts}$ has the form
\begin{equation}\label{eq:100101}
\begin{split}
&(q^s-1)uq^{ts}+q^{ts}-1\\
&=(\underbrace{a_{m-1},a_{m-2},\ldots,a_{m-l},}_{l} \underbrace{0,0,\ldots,0,}_{m-ts-s}\underbrace{a_{ts+s-l-1},a_{ts+s-l-2},\ldots,a_{ts},}_{s-l}\underbrace{q-1,q-1,\ldots,q-1}_{ts})_q
\end{split}
\end{equation}
if there exist $m-ts-s$ consecutive zeros in the $q$-adic expansion of $A$, where $l$ is a nonnegative integer and $l\leq s$. Then
\begin{equation}\label{eq:1001}
(q^s-1)uq^{ts}=a_{m-1}q^{m-1}+\cdots+a_{m-l}q^{m-l}+a_{ts+s-l-1}q^{ts+s-l-1}+\cdots+a_{ts}q^{ts},
\end{equation}
which implies that
\begin{equation}\label{eq:sss}
a_{m-1}q^{m-1}+\cdots+a_{m-l}q^{m-l}+a_{ts+s-l-1}q^{ts+s-l-1}+\cdots+a_{ts}q^{ts}\equiv 0 \pmod {q^s-1}.
\end{equation}
It is easy to see that
\begin{equation}\label{eq:sss01}
\begin{split}
&a_{m-1}q^{m-1}+\cdots+a_{m-l}q^{m-l}+a_{ts+s-l-1}q^{ts+s-l-1}+\cdots+a_{ts}q^{ts}\\
&\equiv a_{m-1}q^{s-1}+\cdots+a_{m-l}q^{s-l}+a_{ts+s-l-1}q^{s-l-1}+\cdots+a_{ts}\pmod {q^s-1}.
\end{split}
\end{equation}
Then, we obtain that (\ref{eq:sss}) and (\ref{eq:sss01}) hold simultaneously if and only if
$$a_{m-1}=a_{m-2}=\cdots=a_{ts}=0\,\,\, \text{ or}\,\,\, a_{m-1}=a_{m-2}=\cdots=a_{ts}=q-1.$$
If $a_{m-1}=a_{m-2}=\cdots=a_{m-s}=0$, from (\ref{eq:1001}) we have $(q^s-1)uq^{ts}=0$, i.e., $u=0$, which is contradictory to $u\geq 1$.
If $a_{m-1}=a_{m-2}=\cdots=a_{ts}=q-1$, from (\ref{eq:100101}) we know that
\begin{equation*}
\begin{split}
(q^s-1)uq^{ts}+q^{ts}-1=(\underbrace{q-1,\ldots,q-1,}_{l} \underbrace{0,0,\ldots,0,}_{m-ts-s}\underbrace{q-1,q-1,\ldots,q-1}_{ts+s-l})_q.
\end{split}
\end{equation*}
From (\ref{eq:10010101}), we have $A=q^{ts+s}-1$, which is contradictory to $A\neq q^{ts+s}-1$. This completes the proof.
\end{proof}

\begin{lemma}\label{lem:2225}
Let $n=\frac{q^m-1}{q^s -1}$ and $\frac{m}{s}\geq 3$. Let $2\leq \delta \leq n$ and $I(\delta)\geq 1$ be the integer such that $\{0, 1, 2, \ldots, I(\delta)-1\} \subseteq  T^{\perp}$ and $I(\delta)\notin  T^{\perp}$. Then
\begin{equation*}
\begin{aligned}
I(\delta)=\begin{cases}
\frac{q^{m-ts}-1}{q^s-1}, &{\rm if}\,\,  \frac{q^{ts}-1}{q^s-1}< \delta \leq \frac{q^{(t+1)s}-1}{q^s-1}\,\,\,(1\leq t \leq \frac{m}{s}-2),\\
1,&{\rm if}\,\, \frac{q^{m-s}-1}{q^s-1} < \delta \leq n.
               \end{cases}
\end{aligned}
\end{equation*}
\end{lemma}

\begin{proof} By the definition of $T^{\perp}$, it is easy to check that $0 \in T^{\perp}$. We now prove that $\{1, 2, \ldots, I(\delta)-1\} \subseteq  T^{\perp}$ and $I(\delta)\notin  T^{\perp}$. There are two cases for discussion.

\noindent {\bf Case 1:} $\frac{q^{ts}-1}{q^s-1}< \delta \leq \frac{q^{(t+1)s}-1}{q^s-1}$, where $1\leq t \leq \frac{m}{s}-2$.
It is easy to see that
\[\frac{q^m-q^{m-ts}}{q^s-1}=\frac{q^{m-ts}(q^{ts}-1)}{q^s-1}\in \mathbb{C}_{\frac{q^{ts}-1}{q^s-1}}\subseteq T.\]
Then we have $\frac{q^{m-ts}-1}{q^s-1}=n-\frac{q^m-q^{m-ts}}{q^s-1} \in T^{-1}$ and $\frac{q^{m-ts}-1}{q^s-1}\notin T^{\perp}$.

In the following, we show that $\{1, 2, \ldots, \frac{q^{m-ts}-1}{q^s-1}-1\}\subseteq T^{\perp}$. For every integer $i$ with $1\leq i\leq \frac{q^{m-ts}-1}{q^s-1}-1$, we have $i=\frac{q^{m-ts}-1}{q^s-1}-u$, where $1\leq u\leq  \frac{q^{m-ts}-1}{q^s-1}-1$. Note that
\[\frac{q^{m-ts}(q^{ts}-1+(q^s-1)uq^{ts})}{q^s-1}\equiv \frac{q^m-1}{q^s-1}-\frac{q^{m-ts}-1}{q^s-1}+u \pmod n.\]
Then
\[\frac{q^{ts}-1+(q^s-1)uq^{ts}}{q^s-1}\in\mathbb{C}_{\frac{q^m-1}{q^s-1}-\frac{q^{m-ts}-1}{q^s-1}+u}= \mathbb{C}_{n-i}.\]

From Lemma \ref{eq:lemt}, we obtain
\[{\rm CL}(n-i)={\rm CL}\left(n-\frac{q^{m-ts}-1}{q^s-1}+u\right) > \frac{q^{ts+s}-1}{q^s-1}-1 \geq \delta-1,\]
where ${\rm CL}\left(n-\frac{q^{m-ts}-1}{q^s-1}+u\right)$ denotes the coset leader of the $q$-cyclotomic coset modulo $n$ containing  $n-\frac{q^{m-ts}-1}{q^s-1}+u$.
Then $n-(\frac{q^{m-ts}-1}{q^s-1}-u)\notin T$ and $\frac{q^{m-ts}-1}{q^s-1}-u\notin T^{-1}$. This leads to $i=\frac{q^{m-ts}-1}{q^s-1}-u\in T^{\perp}$. Thus, we have $I(\delta)=\frac{q^{m-ts}-1}{q^s-1}$.

\noindent {\bf Case 2:} $\frac{q^{m-s}-1}{q^s-1} < \delta \leq n$. It is easy to see that
\[\frac{q^m-1}{q^s-1}-1\equiv \frac{q^{s}(q^{m-s}-1)}{q^s-1}\in \mathcal{C}_{\frac{q^{m-s}-1}{q^s-1}}\subseteq T.\]
Then $1= n-(\frac{q^m-1}{q^s-1}-1)\in T^{-1}$ and $1\notin T^{\perp}$.
 Thus we have $I(\delta)=1$.
We completes the proof.
\end{proof}

From Lemma \ref{lem:2225} and BCH bound, we have the follow theorem on the lower bound on the minimum distance of $\mathcal{C}^{\perp}_{\delta}$.

\begin{theorem}\label{thm:1002}
Let $n=\frac{q^m-1}{q^s-1}$ and $\frac{m}{s}\geq 3$. Let $d^{\perp}(\delta)$ be the minimum distance of $\mathcal{C}^{\perp}_{\delta}$. Then, we have
\begin{equation*}
\begin{aligned}
d^{\perp}(\delta)\geq\begin{cases}
\frac{q^{m-ts}-1}{q^s-1}+1, &{\rm if}\,\,  \frac{q^{ts}-1}{q^s-1}< \delta \leq \frac{q^{(t+1)s}-1}{q^s-1}\,\,\,(1\leq t \leq \frac{m}{s}-2),\\
2,&{\rm if}\,\, \frac{q^{m-s}-1}{q^s-1} < \delta \leq n.
               \end{cases}
\end{aligned}
\end{equation*}
\end{theorem}

\begin{remark}  In \cite[Theorem 28]{GDL21}, the authors gave the lower bound on the minimum distance of $\mathcal{C}^{\perp}_{\delta}$ for $s=1$ and we showed their results  in Lemma \ref{lem:7}. Theorem \ref{thm:1002} generalized their results from the case $s=1$ to the case $s\,|\,m$.
\end{remark}

It is very hard to determine the minimum distance of $\mathcal C_{\delta}^\bot$ in general. The following examples show that the lower bounds in Theorem \ref{thm:1002} are good in some cases.

\begin{example}
Let $\delta=3$, $q=2$, $s=1$, $t=1$ and $m=6$. In Theorem \ref{thm:1002}, the lower bound on the minimum distance of $\mathcal C_{3}^\bot$ is $32$. By Magma, the true minimum distance of $\mathcal C_{3}^\bot$ is $32$.
\end{example}

\begin{example}
Let $\delta=15$, $q=2$,  $s=1$, $t=1$ and $m=6$. In Theorem \ref{thm:1002}, the lower bound on the minimum distance of $\mathcal C_{15}^\bot$ is $8$.  By Magma, the true minimum distance of $\mathcal C_{15}^\bot$ is $8$.
\end{example}

In the following, we present a sufficient and necessary condition for $\mathcal{C}_{\delta}$ being a dually-BCH code. We first give a key lemma.

\begin{lemma}\label{lem:1002}
Let $n=\frac{q^m-1}{q^s-1}$ and $\frac{m}{s}\geq 3$. Let $s>1$, $1\leq t \leq \frac{m}{s}-2$ and $\frac{q^{ts}-1}{q^s-1}<\delta\leq \frac{q^{(t+1)s}-1}{q^s-1}$. Then
 $\frac{q^{m-ts}+q^{m-ts-1}-q^{m-(t+1)s-1}-1}{q^s-1} \in T^{\perp}$ is a $q$-cyclotomic coset leader modulo $n$.

\end{lemma}

\begin{proof}
Note that
\[q^{m-ts}+q^{m-ts-1}-q^{m-(t+1)s-1}-1=
(\underbrace{0,\ldots,0}_{st-1},1,0,\underbrace{q-1,\ldots,q-1}_{s-1},q-2,\underbrace{q-1,\ldots,q-1}_{m-(t+1)s-1})_q.\]
Since $s>1$, it is easy to check that $q^{m-ts}+q^{m-ts-1}-q^{m-(t+1)s-1}-1$ is a $q$-cyclotomic coset leader modulo $q^m-1$. Obviously,
$$q^s-1\, | \, (q^{m-ts}-1)+q^{m-st-s-1}(q^s-1)=q^{m-ts}+q^{m-ts-1}-q^{m-(t+1)s-1}-1.$$
Then, from Lemma \ref{lem:0913} we know that $\frac{q^{m-ts}+q^{m-ts-1}-q^{m-(t+1)s-1}-1}{q^s-1}$ is a $q$-cyclotomic coset leader modulo $n$.

It is easy to see that
\[n-\frac{q^{m-ts}+q^{m-ts-1}-q^{m-(t+1)s-1}-1}{q^s-1}=\frac{q^{m}-q^{m-ts}-q^{m-ts-1}+q^{m-(t+1)s-1}}{q^s-1}\]
and
\begin{equation}\label{eq:qm02}
q^{m}-q^{m-ts}-q^{m-ts-1}+q^{m-(t+1)s-1}=(\underbrace{q-1,\ldots,q-1}_{ts-1},q-2,q-1,\underbrace{0,\ldots,0,}_{s-1}
1,\underbrace{0,\cdots,0}_{m-(t+1)s-1})_q.
\end{equation}
Since $1\leq t \leq \frac{m}{s}-2$, we have $m-ts>2s$. Then, from (\ref{eq:qm02}) we know that the $q$-cyclotomic coset leader of $\mathbb{C}_{q^{m}-q^{m-ts}-q^{m-ts-1}+q^{m-(t+1)s-1}}$ modulo $q^m-1$ is
\[q^{(t+1)s+1}-q^{s+1}-q^s+1=(\underbrace{0,\ldots,0,}_{m-(t+1)s-1}\underbrace{q-1,\ldots,q-1,}_{ts-1}q-2,q-1,\underbrace{0,\ldots,0,}_{s-1}
1)_q.\]
Hence, from Lemma \ref{lem:0913} we obtain
\begin{equation*}
\begin{split}
{\rm CL}\left(n-\frac{q^{m-ts}+q^{m-ts-1}-q^{m-(t+1)s-1}-1}{q^s-1}\right)=\frac{q^{(t+1)s+1}-q^{s+1}-q^s+1}{q^s-1}>\delta-1.
\end{split}
\end{equation*}
This leads to $n-\frac{q^{m-ts}+q^{m-ts-1}-q^{m-(t+1)s-1}-1}{q^s-1}\notin T$ and $\frac{q^{m-ts}+q^{m-ts-1}-q^{m-(t+1)s-1}-1}{q^s-1}\notin T^{-1}$. Thus we have $\frac{q^{m-ts}+q^{m-ts-1}-q^{m-(t+1)s-1}-1}{q^s-1}\in T^{\perp}$. This completes the proof.
\end{proof}
%
%\noindent {\bf Case 2:} $\frac{q^{m-2s}-1}{q^s-1} < \delta \leq \frac{q^{m-s}-1}{q^s-1}$.  Note that
%\[q^{2s}+q^{s+1}-q-1=
%(\underbrace{0,\ldots,0}_{m-2s-1},1,\underbrace{0,\ldots,0}_{s-1},\underbrace{q-1,\ldots,q-1}_{s-1},q-2,q-1)_q.\]
%Since $m\geq 3s$, we can check that $q^{2s}+q^{s+1}-q-1$ is a $q$-cyclotomic coset leader modulo $q^m-1$.
%Then from Lemma \ref{lem:0913}, we know that $\frac{q^{2s}+q^{s+1}-q-1}{q^s-1}$ is a $q$-cyclotomic coset leader modulo $n$.
%
%Obviously, we have
%\[n-\frac{q^{2s}+q^{s+1}-q-1}{q^s-1}=\frac{q^{m}-q^{2s}-q^{s+1}+q}{q^s-1}\]
%and
%\[q^{m}-q^{2s}-q^{s+1}+q=(\underbrace{q-1,\ldots,q-1}_{m-2s-1},q-2,\underbrace{q-1,\ldots,q-1}_{s-1},
%\underbrace{0,\ldots,0}_{s-1},1,0)_q.\]
%With an analysis similar to Case 1, the $q$-cyclotomic coset leader containing $q^{m}-q^{2s}-q^{s+1}+q$ modulo $q^m-1$ is
%\[q^{(t+1)s+1}-q^{s+1}-q^s+1=(\underbrace{0,\ldots,0,}_{m-(t+1)s-1}\underbrace{q-1,\ldots,q-1}_{ts-1},q-2,q-1,\underbrace{0,\ldots,0,}_{s-1}
%1)_q\]
%and $\frac{q^{2s}+q^{s+1}-q-1}{q^s-1}\in T^{\perp}.$

From Lemma \ref{lem:1002} and the previous conclusions, we can get the following results.

\begin{theorem}\label{thm-02} Let
$n=\frac{q^m-1}{q^s-1}$ be the length of BCH code $\mathcal{C}_\delta$, where $\frac{m}{s}\geq 3$. Let $m\geq 4$ if $q\neq 2$ and $m\geq 6$ if $q=2$. If $q=2$ and $s=1$, then $\mathcal{C}_{\delta}$ is a dually-BCH code if and only if
$$\delta=2,3, or \,\,\,\,\delta_2+1 \leq \delta \leq n,$$
where $\delta_2$ is given in Lemma \ref{lem:qm1q1}.
If $q> 2$ or $s>1$, then $\mathcal{C}_{\delta}$ is a dually-BCH code if and only if
$$ \delta_1< \delta \leq n,$$
where $\delta_1$ is the first largest $q$-cyclotomic coset leader modulo $n$.

\end{theorem}

\begin{proof}
If $s=1$, the results have been given in  \cite[Theorem 17]{GDL21} for $q=2$, \cite[Theorem 30]{GDL21} for $q=3$ and \cite[Theorem 4.13]{Wang23} for $q>3$.
In the following, we always assume that $s>1$.

It is clear that $0\notin T$ and $1\in T$, then $0 \notin T^{-1}$ and $n-1 \in T^{-1}$, which means that $0 \in T^{\perp}$ and $n-1 \notin T^{\perp}$. Hence, $\mathbb{C}_0$ must be the initial cyclotomic coset of $T^{\perp}$. In other words, there must be an integer $J\geq 1$ such that $T^{\perp}=\mathbb{C}_0\cup \mathbb{C}_1\cup \mathbb{C}_2 \cup \cdots \cup \mathbb{C}_{J-1}$ if $\mathcal{C}_{\delta}$ is a dually-BCH code.

If $\delta_1< \delta \leq n$, it is easy to see that $T^{\perp}=\{0\}$ and $\mathcal{C^{\perp}_{\delta}}$ is a BCH code with respect to $\beta$.
It remains to show that $\mathcal{C^{\perp}_{\delta}}$ is not a dually-BCH code with respect to $\beta$ when $2\leq \delta\ \leq \delta_1$ for $m\geq 3s$. To this end, we show that there doesn't exist an integer $J\geq 1$ such that $T^{\perp}=\mathbb{C}_0\cup \mathbb{C}_1\cup \mathbb{C}_2 \cup \cdots \cup \mathbb{C}_{J-1}$. The discussion is divided into the following two cases.

\noindent{\bf Case 1:} $\frac{q^{ts}-1}{q^s-1}<\delta\leq \frac{q^{(t+1)s}-1}{q^s-1}$, where $1\leq t \leq \frac{m}{s}-2$. In this case, from Lemma \ref{lem:1002}, we know that $\frac{q^{m-ts}+q^{m-ts-1}-q^{m-(t+1)s-1}-1}{q^s-1} \in T^{\perp}$ is a $q$-cyclotomic coset leader modulo $n$. It then follows from Lemma \ref{lem:2225}, we have
$$I_{max}\,:=max\left\{I(\delta)\,:\, \frac{q^{ts}-1}{q^s-1}<\delta\leq \frac{q^{(t+1)s}-1}{q^s-1}\right\}=\frac{q^{m-ts}-1}{q^s-1}.$$
It is clear that $\frac{q^{m-ts}+q^{m-ts-1}-q^{m-(t+1)s-1}-1}{q^s-1} > \frac{q^{m-ts}-1}{q^s-1}$. Thus, there is no an integer $J\geq 1$ such that $T^{\perp}=\mathbb{C}_0\cup \mathbb{C}_1\cup \mathbb{C}_2 \cup \cdots \cup \mathbb{C}_{J-1}$, i.e., $\mathcal{C^{\perp}_{\delta}}$ is not a BCH code with respect to $\beta.$

\noindent{\bf Case 2:}  $\frac{q^{m-s}-1}{q^s-1} < \delta \leq \delta_1$. From Lemma  \ref{lem:2225}, we know that
$$I_{max}\,:=max\left\{I(\delta)\,:\, \frac{q^{m-s}-1}{q^s-1} < \delta \leq \delta_1\right\}=1,$$
which means that $T^{\perp}=\{0\}$ if $\mathcal{C^{\perp}_{\delta}}$ is a BCH code with respect to $\beta$.  By definition, we have $\delta_1 \notin T$, then $n-\delta_1 \notin T^{-1}$ and $n-\delta_1 \in T^{\perp}$. This leads to
$${\rm CL}(n-\delta_1)\in T^{\perp}.$$
It is obvious that ${\rm CL}(n-\delta_1)\neq 0$. Then there must be a $q$-cyclotomic coset leader modulo $n$, which is greater than $1$ in $T^{\perp}$. Thus, there doesn't exist an integer $J\geq 1$ such that $T^{\perp}=\mathbb{C}_0\cup \mathbb{C}_1\cup \mathbb{C}_2 \cup \cdots \cup \mathbb{C}_{J-1}$, i.e., $\mathcal{C^{\perp}_{\delta}}$ is not a BCH code with respect to $\beta.$
The desired conclusion follows.
\end{proof}

\begin{example}
Let $q=3$, $s=2$ and $m=6$. By Magma, $\delta_1=49$  and $\mathcal{C}_\delta$ is a dually-BCH code if and only if $49 < \delta\leq 91$, where $\delta_1$ is the first largest  $q$-cyclotomic coset leader modulo $91$. This experimental result coincides with Theorem \ref{thm-02}.
\end{example}

\begin{example}
Let $q=3$, $s=3$ and $m=9$. By Magma, $\delta_1=388$  and $\mathcal{C}_\delta$ is a dually-BCH code if and only if $388 < \delta\leq 757$,  where $\delta_1$ is the first largest  $q$-cyclotomic coset leader modulo $757$. This experimental result coincides with Theorem \ref{thm-02}.
\end{example}

\begin{remark}
It seems difficult  to determine the first largest $q$-cyclotomic coset leader modulo $\frac{q^m-1}{q^s-1}$ in general. When $s=1$, the first few largest $q$-cyclotomic coset leaders were given in Lemma \ref{lem:qm1q1} and Lemma \ref{lem:qm1q101} if $q=2$ and $q\geq 3$, respectively.
\end{remark}

%
%If $s=m$, it is clear that the code $\mathcal{C}^{\perp}_{(q,n,\delta)}$ does not have any meaning. If $m=2s$, then the In the following, we always assume that $m\geq 3s$.

\subsection{The case $\lambda\mid q-1$ }

In this subsection, we always assume that $\lambda\neq q-1$ since the results for $\lambda=q-1$ have been given in Subsection $A$. We first give the lower bound on the minimum distance of the dual of $\mathcal{C}_{\delta}$. To this end, the following lemma will be employed later.

\begin{lemma}\label{lem:0916}
Let $1\leq s\leq \frac{q-1}{\lambda}-1$ and $0\leq t\leq m-2$ be integers. Let $1\leq u\leq  \frac{q^{m-t}-1}{\lambda}-sq^{m-t-1}-1$, then the coset leader in  $\mathbb{C}_{(\lambda u+1)q^{t+1}-q+\lambda s}$ modulo $q^m-1$ is greater than  $q^{t+1}-q+\lambda s$.
\end{lemma}

\begin{proof} It is easy to check that $(\lambda u+1)q^{t+1}-q+\lambda s\neq q^m-1$. We now prove this lemma from the following two cases.

\noindent {\bf Case 1:}  $0<(\lambda u+1)q^{t+1}-q+\lambda s<q^m-1$. By the definition of $\lambda$ and $s$, we know that $\lambda s < q-1$. Then it is straightforward to see that
\begin{equation}\label{eq:0913}
q^{t+1}-q+\lambda s=(\underbrace{0,\ldots,0}_{m-t-1},\underbrace{q-1,\ldots,q-1}_{t},\lambda s)_q
\end{equation}
and
\begin{equation}\label{eq:091301}
(\lambda u+1)q^{t+1}-q+\lambda s=\lambda uq^{t+1}+(q^{t+1}-q+\lambda s) = (\underbrace{i_{m-1},\ldots,i_{t+1}}_{m-t-1}, \underbrace{q-1,\ldots,q-1}_{t},\lambda s)_q.
\end{equation}
By the definition of $\lambda$ and $u$, there  exist at least two integers $i_{j_1}, i_{j_2}$ in (\ref{eq:091301}) such that $i_{j_1}\neq 0$ and $i_{j_2}\neq 0$, where $t+1\leq j_1,j_2 \leq m-1$. From (\ref{eq:0913}), (\ref{eq:091301}) and Lemma \ref{lem1b21}, the coset leader of $\mathbb{C}_{(\lambda u+1)q^{t+1}-q+\lambda s}$ modulo $q^m-1$ is greater than  $q^{t+1}-q+\lambda s$.

\noindent {\bf Case 2:} $(\lambda u+1)q^{t+1}-q+\lambda s>q^m-1$. In this case, it is obvious that $\lambda u\geq q^{m-t-1}$.
By the definition of $\lambda$ and $u$, we have $\lambda u\leq q^{m-t}-\lambda s q^{m-t-1}-\lambda-1$. Then $\lambda u$ can be expressed as
$$\lambda u=a_{m-t-1}q^{m-t-1}+a_{m-t-2}q^{m-t-2}+\cdots+a_1q+a_0,$$
where $0\leq a_0,\ldots,a_{m-t-2}\leq q-1$ and $0\leq a_{m-t-1}\leq q-\lambda s-1$.
Hence, the $q$-adic expansion of $(\lambda u+1)q^{t+1}-q+\lambda s$ is
\begin{equation}\label{eq:1025}
\begin{split}
(\lambda u+1)q^{t+1}-q+\lambda s&=\lambda uq^{t+1}+(q^{t+1}-q+\lambda s) \\
&\equiv(\underbrace{a_{m-t-2},\ldots,a_0}_{m-t-2}, \underbrace{q-1,\ldots,q-1}_{t},\lambda s,a_{m-t-1})_q \pmod {q^m-1}.
\end{split}
\end{equation}
From (\ref{eq:0913}), (\ref{eq:1025}) and Lemma \ref{lem1b21}, the coset leader of $\mathbb{C}_{(\lambda u+1)q^{t+1}-q+\lambda s}$ modulo $q^m-1$ is greater than  $q^{t+1}-q+\lambda s$.
The desired conclusion then follows.
\end{proof}

\begin{lemma}\label{eq:lem919}
Let $n=\frac{q^m-1}{\lambda}$ and $2\leq \delta \leq n$. Let $I(\delta)\geq 1$ be the integer such that $\{0, 1, 2, \ldots, I(\delta)-1\} \subseteq T^{\perp}$ and $I(\delta)\notin  T^{\perp}$. Then
\begin{equation*}
\begin{aligned}
I(\delta)=\begin{cases}
\frac{q^{m-t}-1}{\lambda}-sq^{m-t-1}, &{\rm if} \,\,\,  \delta = \frac{q^{t+1}-q}{\lambda}+s+1\,\, (1\leq s\leq \frac{q-1}{\lambda}-1,\,\, 0\leq t\leq m-2),\\
\frac{q^{m-t}-1}{\lambda}-s, &{\rm if} \,\,\, \frac{q^t-1}{\lambda}+sq^t < \delta \leq \frac{q^t-1}{\lambda}+(s+1)q^t\,\,(0\leq s\leq \frac{q-1}{\lambda}-2,\,\, 1\leq t\leq m-1),\\
\frac{q^{m-t}-q}{\lambda}+1, &{\rm if} \,\,\, \frac{q^{t+1}-1}{\lambda}-q^t < \delta \leq \frac{q^{t+1}-q}{\lambda}+1\,\,(1\leq t\leq m-2),\\
1,&{\rm if} \,\,\, \frac{q^m-1}{\lambda}-q^{m-1} < \delta \leq n.
               \end{cases}
\end{aligned}
\end{equation*}
\end{lemma}

\begin{proof}
By the definition of $T^{\perp}$, it is easy to check that $0 \in \ T^{\perp}$. In the following, we prove that $\{1, 2, \ldots, I(\delta)-1\} \subseteq T^{\perp}$ and $I(\delta)\notin  T^{\perp}$. There are four cases for discussion.

\noindent {\bf Case 1:} $\delta = \frac{q^{t+1}-q}{\lambda}+s+1$, where $1\leq s\leq \frac{q-1}{\lambda}-1$ and $0\leq t\leq m-2$. It is straightforward to see that
\[ n-\frac{q^{m-t}-1}{\lambda}+sq^{m-t-1} = \frac{q^{m-t-1}(q^{t+1}-q+\lambda s)}{\lambda}\in \mathbb{C}_{\frac{q^{t+1}-q}{\lambda}+s}\subseteq T.\]
Then $n-(n-\frac{q^{m-t}-1}{\lambda}+sq^{m-t-1})=\frac{q^{m-t}-1}{\lambda}-sq^{m-t-1} \in T^{-1}$ and $\frac{q^{m-t}-1}{\lambda}-sq^{m-t-1}\notin T^{\perp}$.

 For every integer $i$ with $1\leq i\leq \frac{q^{m-t}-1}{\lambda}-sq^{m-t-1}-1$, it is clear that $i$ can be expressed as $i=\frac{q^{m-t}-1}{\lambda}-sq^{m-t-1}-u$, where $1\leq u\leq  \frac{q^{m-t}-1}{\lambda}-sq^{m-t-1}-1$. Note that
\[\frac{q^{m-t-1}(\lambda uq^{t+1}+q^{t+1}-q+\lambda s)}{\lambda}\equiv n-\frac{q^{m-t}-1}{\lambda}+sq^{m-t-1}+u \pmod n,\]
then
\[\frac{\lambda uq^{t+1}+q^{t+1}-q+\lambda s}{\lambda}\in\mathbb{C}_{n-\frac{q^{m-t}-1}{\lambda}+sq^{m-t-1}+u}= \mathbb{C}_{n-i}.\]

Hence, from Lemma \ref{lem:0913} and Lemma \ref{lem:0916}, we know that
\[{\rm CL}(n-i)={\rm CL}\left(n-\frac{q^{m-t}-1}{\lambda}+sq^{m-t-1}+u\right) > \frac{q^{t+1}-q}{\lambda}+s\geq \delta-1,\]
where ${\rm CL}\left(n-\frac{q^{m-t}-1}{\lambda}+sq^{m-t-1}+u\right)$ denotes the coset leader of the $q$-cyclotomic coset modulo $n$ containing $n-\frac{q^{m-t}-1}{\lambda}+sq^{m-t-1}+u$. Consequently, $n-(\frac{q^{m-t}-1}{\lambda}-sq^{m-t-1}-u)\notin T$ and $\frac{q^{m-t}-1}{\lambda}-sq^{m-t-1}-u\notin T^{-1}$. This leads to $i=\frac{q^{m-t}-1}{\lambda}-sq^{m-t-1}-u\in T^{\perp}$.

Therefore, we obtain that $\{0, 1, 2, \ldots, \frac{q^{m-t}-1}{\lambda}-sq^{m-t-1}-1\}\subseteq T^{\perp}$ and $\frac{q^{m-t}-1}{\lambda}-sq^{m-t-1}\notin T^{\perp}$, i.e., $I(\delta)=\frac{q^{m-t}-1}{\lambda}-sq^{m-t-1}$.

\noindent {\bf Case 2:} $\frac{q^t-1}{\lambda}+sq^t < \delta \leq \frac{q^t-1}{\lambda}+(s+1)q^t$, where $0\leq s\leq \frac{q-1}{\lambda}-2$ and $1\leq t\leq m-1$. It is easy to see that
\[\frac{(\lambda s+1)q^m-q^{m-t}}{\lambda}
 \equiv n-\frac{q^{m-t}-1}{\lambda}+s \pmod n,\]
then
\[n-\frac{q^{m-t}-1}{\lambda}+s\in \mathbb{C}_{\frac{(\lambda s+1)q^t-1}{\lambda}}\subseteq T \]
since $\frac{(\lambda s+1)q^m-q^{m-t}}{\lambda}=\frac{q^{m-t}((\lambda s+1)q^t-1)}{\lambda}.$
Hence, $n-(n-\frac{q^{m-t}-1}{\lambda}+s)=\frac{q^{m-t}-1}{\lambda}-s\in T^{-1}$ and $\frac{q^{m-t}-1}{\lambda}-s\notin T^{\perp}$.

We next to show that $\{1, 2, \ldots, \frac{q^{m-t}-1}{\lambda}-s-1\}\subseteq T^{\perp}$. For every integer $i$ with $1\leq i\leq \frac{q^{m-t}-1}{\lambda}-s-1$, we can assume that $i=\frac{q^{m-t}-1}{\lambda}-s-u$, where $1\leq u\leq  \frac{q^{m-t}-1}{\lambda}-s-1$. Note that
\[\frac{q^{m-t}(\lambda(s+u)q^t+q^t-1)}{\lambda}\equiv \frac{q^m-q^{m-t}+\lambda(s+u)}{\lambda} \pmod n,\]
then
$\frac{\lambda(s+u)q^t+q^t-1}{\lambda}\in \mathbb{C}_{\frac{q^m-q^{m-t}+\lambda(s+u)}{\lambda}}=\mathbb{C}_{n-i}$.

By the definition of $s$ and $u$, we have $\lambda(s+u)q^t+q^t-1<q^m-1$. Then the $q$-adic expansion of $\lambda(s+u)q^t+q^t-1$ can be expressed as
\begin{equation}\label{eq:ambda}
\lambda(s+u)q^t+q^t-1=(\underbrace{0,\ldots,0,\lambda(s+u)}_{m-t}, \underbrace{q-1,\ldots,q-1}_{t})_q
\end{equation}
if $\lambda(s+u)<q$, and
\begin{equation}\label{eq:ambda1}
\lambda(s+u)q^t+q^t-1=(\underbrace{i_{m-1},i_{m-2},\ldots,i_{t}}_{m-t}, \underbrace{q-1,\ldots,q-1}_{t})_q
\end{equation}
if $\lambda(s+u)\geq q$, where $i_{j}\neq 0$ for some integers $t\leq j \leq m-1$.

Since $\lambda(s+1)\leq q-1$, it is easy to see that
$$\lambda(s+1)q^t+q^t-1=(\underbrace{0,\ldots,0,\lambda(s+1)}_{m-t}, \underbrace{q-1,\ldots,q-1}_{t})_q$$
and $\lambda(s+1)q^t+q^t-1$ is a $q$-cyclotomic coset leader modulo $q^m-1$.
From (\ref{eq:ambda}) and (\ref{eq:ambda1}) the coset leader of $\mathbb{C}_{\lambda(s+u)q^t+q^t-1}$ modulo $q^m-1$ is greater than or equal to ${\lambda(s+1)q^t+q^t-1}$.
Then, from Lemma \ref{lem:0913} we have
\[{\rm CL}(n-i)={\rm CL}\left(\frac{q^m-q^{m-t}+\lambda(s+u)}{\lambda}\right) \geq \frac{q^t-1}{\lambda}+(s+1)q^t> \delta-1.\]
Hence, $n-(\frac{q^{m-t}-1}{\lambda}-u-s)\notin T$ and $\frac{q^{m-t}-1}{\lambda}-u-s\notin T^{-1}$. This leads to $i=\frac{q^{m-t}-1}{\lambda}-u-s\in T^{\perp}$. Thus, we obtain $I(\delta)=\frac{q^{m-t}-1}{\lambda}-s$.

\noindent {\bf Case 3:} $\frac{q^{t+1}-1}{\lambda}-q^t < \delta \leq \frac{q^{t+1}-q}{\lambda}+1$, where $1\leq t\leq m-2$. Obviously, we have
\[n-\frac{q^{m-t}-q+\lambda}{\lambda}\equiv \frac{q^{m-t}(q^{t+1}-\lambda q^t-1)}{\lambda} \pmod n, \]
then
\[n-\frac{q^{m-t}-q+\lambda}{\lambda}\in \mathbb{C}_{\frac{q^{t+1}-\lambda q^t-1}{\lambda}}\subseteq T. \]
By definition, we obtain that $\frac{q^{m-t}-q+\lambda}{\lambda} \in T^{-1}$ and $\frac{q^{m-t}-q+\lambda}{\lambda} \notin T^{\perp}$.

In the following, we are going to show that $\{0, 1, 2, \ldots, \frac{q^{m-t}-q}{\lambda}\}\subseteq T^{\perp}$. For every integer $i$ with $1\leq i\leq \frac{q^{m-t}-q}{\lambda}$, we have $i=\frac{q^{m-t}-q}{\lambda}-u$, where $0\leq u\leq  \frac{q^{m-t}-q}{\lambda}-1$. Note that
\[\frac{q^{m-t}(q^{t+1}+\lambda uq^t-1)}{\lambda}\equiv n-\frac{q^{m-t}-q}{\lambda}+u \pmod n,\]
then $\frac{q^{t+1}+\lambda uq^{t}-1}{\lambda}\in \mathbb{C}_{n-\frac{q^{m-t}-q}{\lambda}+u}=\mathbb{C}_{n-i}$. In addition, we have
\[q^{t+1}+\lambda uq^{t}-1=(q+\lambda u-1)q^{t}+q^t-1= (\underbrace{i_{m-1},i_{m-2},\ldots,i_{t+1},i_{t}}_{m-t}, \underbrace{q-1,\ldots,q-1}_{t})_q.\]
It is clear that $i_t=q-1$ if $u=0$, and  there exist at least two integers $ i_{j_1}\neq 0$ and $i_{j_2}\neq 0$ if $u\neq 0$, where $t\leq j_1,\,\,j_2 \leq m-1$.

Since
\[q^{t+1}-q+\lambda= (\underbrace{0,0,\ldots,0,0}_{m-t-1}, \underbrace{q-1,q-1,\ldots,q-1}_{t},\lambda),\]
we know that the coset leader of $\mathbb{C}_{q^{t+1}+\lambda uq^{t}-1}$ modulo $q^m-1$ is greater than $q^{t+1}-q+\lambda$.
Then, from Lemma \ref{lem:0913} we have
\[{\rm CL}(n-i)={\rm CL}\left(\frac{q^m-1}{\lambda}-\frac{q^{m-t}-q}{\lambda}+u\right)> \frac{q^{t+1}-q+\lambda}{\lambda}> \delta-1.\]
Hence, $n-\frac{q^{m-t}-q}{\lambda}+u\notin T$ and $\frac{q^{m-t}-q}{\lambda}-u\notin T^{-1}$. This leads to $i=\frac{q^{m-t}-q}{\lambda}-u\in T^{\perp}$ for any $0\leq u\leq  \frac{q^{m-t}-q}{\lambda}-1$. Thus, we have $I(\delta)=\frac{q^{m-t}-q}{\lambda}+1$.

\noindent {\bf Case 4:} $\frac{q^m-1}{\lambda}-q^{m-1} < \delta \leq \delta_1$. It is easy to see that
\[\frac{q^m-1}{\lambda}-1\equiv \frac{q^{m-1}(q^m-\lambda q^{m-1}-1)}{\lambda} \pmod n,\]
then
\[\frac{q^m-1}{\lambda}-1 \in \mathbb{C}_{\frac{q^m-\lambda q^{m-1}-1}{\lambda}}\subseteq T.\]
Thus, $1= n-(\frac{q^m-1}{\lambda}-1)\in T^{-1}$ and $1\notin T^{\perp}$.
Note that $0\in T^{\perp}$, then we have $I(\delta)=1$.

Combining all the cases, the desired conclusion follows.
\end{proof}

From BCH bound and Lemma \ref{eq:lem919}, we have the following theorem on the lower bound on the minimum distance of $\mathcal{C}^{\perp}_{\delta}$.
\begin{theorem}\label{eq:thm9190}
Let $n=\frac{q^m-1}{\lambda}$, $\lambda \, | \, q-1$ and $\lambda\neq q-1$. Let $d^{\perp}(\delta)$ be the minimum distance of $\mathcal{C}^{\perp}_{\delta}$. Then we have
\begin{equation*}
\begin{aligned}
d^{\perp}(\delta)\geq
\begin{cases}
\frac{q^{m-t}+\lambda-1}{\lambda}-sq^{m-t-1}, &{\rm if} \,\,\,  \delta = \frac{q^{t+1}-q}{\lambda}+s+1\,\, (1\leq s\leq \frac{q-1}{\lambda}-1,\,\, 0\leq t\leq m-2),\\
\frac{q^{m-t}-1}{\lambda}-s+1, &{\rm if} \,\,\, \frac{q^t-1}{\lambda}+sq^t < \delta \leq \frac{q^t-1}{\lambda}+(s+1)q^t\,\,(0\leq s\leq \frac{q-1}{\lambda}-2,\\
&\,\,\,\,\,\,\, 1\leq t\leq m-1),\\
\frac{q^{m-t}-q}{\lambda}+2, &{\rm if} \,\,\, \frac{q^{t+1}-1}{\lambda}-q^t < \delta \leq \frac{q^{t+1}-q}{\lambda}+1\,\,(1\leq t\leq m-2),\\
2,&{\rm if} \,\, \frac{q^m-1}{\lambda}-q^{m-1} < \delta \leq n.
               \end{cases}
\end{aligned}
\end{equation*}
\end{theorem}

\begin{remark}  In \cite[Theorem 18 and Theorem 23]{GDL21}, the authors gave the lower bound on the minimum distance of $\mathcal{C}^{\perp}_{\delta}$ for $\lambda=1$ and we showed their results in Lemma \ref{lem:8}. Theorem \ref{eq:thm9190} generalized their results from the case $\lambda=1$ to the case $\lambda\, |\, q-1$.
\end{remark}

The following examples show that the lower bounds in Theorem \ref{eq:thm9190} are good in some cases.

\begin{example}
Let $\delta=5$, $q=3$, $\lambda=1$ and $m=3$. In Theorem \ref{eq:thm9190}, the lower bound on the minimum distance of $\mathcal C_{5}^\bot$ is $9$. By Magma, the true minimum distance of $\mathcal C_{5}^\bot$ is $9$.
\end{example}

\begin{example}
Let $\delta=3$, $q=5$, $\lambda=1$ and $m=2$. In Theorem \ref{eq:thm9190}, the lower bound on the minimum distance of $\mathcal C_{3}^\bot$ is $15$.  By Magma, the true minimum distance of $\mathcal C_{3}^\bot$ is $16$.
\end{example}

To present the sufficient and necessary condition for $\mathcal{C}_{\delta}$ being a dually-BCH code, the following several lemmas will be needed later.  By the same way as \cite[Lemma 4.12]{Wang23}, we have the following results.

\begin{lemma} \label{proposition-(q^m-1)/(q-1)}
Let $\delta'$ be the coset leader of $\mathbb{C}_{n-\delta_1}$ modulo $n$, then
\begin{enumerate}
\item $\delta_1 \in T^\bot$ is a $q$-cyclotomic coset leader modulo $n$ if $2 \le \delta \le \delta'$.
\item $\delta' \in T^\bot$ is a $q$-cyclotomic coset leader modulo $n$ if $\delta' < \delta \le \delta_1$.
\end{enumerate}
\end{lemma}

\begin{lemma}\label{eq:0918}
Let $q>3$, $1<\lambda<q-1$, $\lambda\, |\,q-1$ and $n=\frac{q^m-1}{\lambda}$. Then the following statements hold.
\begin{itemize}

\item  $\frac{q^{m}-((\lambda-1)q+1)q^{m-2}-1}{\lambda}\in T^{\perp}$ is a $q$-cyclotomic coset leader modulo $n$ if  $2\leq \delta \leq \frac {q-1}{\lambda}+1$.

\item  Let $0\leq b<\lambda$ satisfy $m+b\equiv 0 \pmod\lambda$. Then $\frac{q^{m-1}+\cdots+q^b+2q^{b-1}+\cdots+2}{\lambda}\in T^{\perp}$ is a $q$-cyclotomic coset leader modulo $n$ if $\frac{q-1}{\lambda} +1< \delta \leq \frac{q^{m-1}-1}{\lambda}$.

\item   $\frac{q-1+\lambda}{\lambda}\in T^{\perp}$ is a  $q$-cyclotomic coset leader modulo $n$ if $\frac {q^{m-1}-1}{\lambda} <\delta\leq n-q^{m-1}$.

\end{itemize}
\end{lemma}
\begin{proof}
We proof these results from the following cases.

\noindent {\bf Case 1:} $2\leq \delta \leq \frac {q-1}{\lambda}+1$. It is easy to check that
$$\lambda \, | \, \left( q^m-1-((\lambda-1)q+1)q^{m-2}\right). $$
Obviously, the $q$-adic expansion of $q^m-1-((\lambda-1)q+1)q^{m-2}$ can be expressed as
$$q^m-1-((\lambda-1)q+1)q^{m-2}=(q-\lambda,q- 2,\underbrace{q-1,\ldots,q-1}_{m-2})_q$$
and
$q^m-1-((\lambda-1)q+1)q^{m-2}$ is a $q$-cyclotomic coset leader modulo $q^m-1$ since $2\leq\lambda<q-1$. Then from Lemma \ref{lem:0913},
 $\frac{q^m-1-((\lambda-1)q+1)q^{m-2}}{\lambda}$ is a $q$-cyclotomic coset leader modulo $n$.

It is obvious that
\begin{equation*}
\begin{split}
{\rm CL}\left(n-\frac{q^m-1-((\lambda-1)q+1)q^{m-2}}{\lambda}\right) & ={\rm CL}\left(\frac {(\lambda-1)q+1}{\lambda}\right)>\delta-1\\
\end{split}
\end{equation*}
since $\delta<q-1$.
It then follows that $n-\frac{q^m-1-((\lambda-1)q+1)q^{m-2}}{\lambda}\notin T$ and $\frac{q^m-1-((\lambda-1)q+1])q^{m-2}}{\lambda}\notin T^{-1}$.
This lead to $\frac{q^m-1-((\lambda-1)q+1)q^{m-2}}{\lambda}\in T^{\perp}$.

\noindent {\bf Case 2:} $\frac{q-1}{\lambda} +1< \delta \leq \frac{q^{m-1}-1}{\lambda}$. It is obvious that $$q^{m-1}+\cdots+q^b+2q^{b-1}+\cdots+2=(\underbrace{1,1,\ldots,1,1}_{m-b},\underbrace{2,2,\ldots,2,2}_{b})_q$$ is a  $q$-cyclotomic  coset leader modulo $q^m-1$.
By the definition of $\lambda$, we know that $\lambda \, | \, q-1$.  Since
$$q^{m-1}+\cdots+q^b+2q^{b-1}+\cdots+2=(q^{m-1}-1)+\cdots+(q^b-1)+2(q^{b-1}-1)+\cdots+2(q-1)+\lambda a,$$ we know that $\lambda \, | \, q^{m-1}+\cdots+q^b+2q^{b-1}+\cdots+2.$
According to Lemma \ref{lem:0913}, $\frac{q^{m-1}+\cdots+q^b+2q^{b-1}+\cdots+2}{\lambda}$ is a $q$-cyclotomic coset leader modulo $n$.

In the following, we prove that $\frac{q^{m-1}+\cdots+q^b+2q^{b-1}+\cdots+2}{\lambda} \in T^{\perp}$.
Note that
$$ q^{m}-q^{m-1}-\cdots-q^b-2q^{b-1}-\cdots-3=(\underbrace{q-2,\ldots,q-2}_{m-b},\underbrace{q-3,\ldots,q-3}_{b})_q. $$
Then, the coset leader of $\mathbb{C}_{q^{m}-q^{m-1}-\cdots-q^b-2q^{b-1}-\cdots-3}$ is
$$q^m-2q^{m-1}-\cdots-2q^{m-b}-q^{m-b-1}-\cdots-q-2=(\underbrace{q-3,\ldots,q-3}_{b},\underbrace{q-2,\ldots,q-2}_{m-b})_q. $$
Hence, we have
\begin{equation*}
\begin{split}
{\rm CL}\left(n-\frac{q^{m-1}+\cdots+q^b+2q^{b-1}+\cdots+2}{\lambda}\right) & ={\rm CL}\left(\frac {q^{m}-q^{m-1}-\cdots-q^b-2q^{b-1}-\cdots-3}{\lambda}\right)\\
&=\frac {q^m-2q^{m-1}-\cdots-2q^{m-b}-q^{m-b-1}-\cdots-q-2}{\lambda}\\
& > \frac{q^{m-1}-1}{\lambda} > \delta-1.
\end{split}
\end{equation*}
Consequently, $n-\frac{q^{m-1}+\cdots+q^b+2q^{b-1}+\cdots+2}{\lambda}\notin T$ and $\frac{q^{m-1}+\cdots+q^b+2q^{b-1}+\cdots+2}{\lambda}\notin T^{-1}$.
This lead to $\frac{q^{m-1}+\cdots+q^b+2q^{b-1}+\cdots+2}{\lambda}\in T^{\perp}$.

\noindent {\bf Case 3:} $\frac {q^{m-1}-1}{\lambda} <\delta\leq n-q^{m-1}$. It is easy to check that $\frac {q-1+\lambda}{\lambda}$ is a $q$-cyclotomic coset leader mod $n$ due to $q-1+\lambda$ is a  $q$-cyclotomic coset leader mod $q^m-1$. Note that
$$q^m-q-\lambda =(q-1,q-1,\ldots,q-2,q-\lambda).$$
Then, we have
$${\rm CL}(q^m-q-\lambda)=q^m-(\lambda-1)q^{m-1}-2=(q-\lambda,q-1,\ldots,q-1,q-2)$$
if $\lambda>2$ and
$${\rm CL}(q^m-q-\lambda)=q^m-q^{m-1}-q^{m-2}-1=(q-2,q-2,\ldots,q-1,q-1)$$
if $\lambda=2$.
With an analysis similar as Case $1$ and Case $2$, it then follows that $n-\frac{q-1+\lambda}{\lambda}\notin T$ and $\frac {q-1+\lambda}{\lambda}\in T^{\perp}$.
This completes the proof.
\end{proof}

The following theorem gives a sufficient and necessary condition for $\mathcal{C}_{\delta}$ being a dually-BCH code, where $2\leq \delta\leq n$.

\begin{theorem}\label{thm-04}
Let $m\geq 2$, $q\geq 3$, $\lambda \, | \, q-1$ and $\lambda\neq q-1$. Let $n=\frac{q^m-1}{\lambda}$ be the length of $\mathcal{C}_{\delta}$. Then $\mathcal{C}_{\delta}$ is a dually-BCH code if and only if
$$\delta_1< \delta \leq n$$
if $\lambda\neq 1$, and
$$\delta=2, or\,\,\,\, \delta_2< \delta \leq n$$
if $\lambda=1$, where $\delta_1$ and $\delta_2$ are the first and the second largest $q$-cyclotomic coset leader modulo $n$, respectively.
\end{theorem}

\begin{proof}
If $\lambda=1$, the results have been given in \cite[Theorem 22]{GDL21}. If $q=3$, by the definition of $\lambda$, we know that $\lambda=1$ since $\lambda\neq q-1$. Hence, in the following, we always assume that $\lambda\neq1$ and $q>3$.

It is clear that $0\notin T$ and $1\in T$, so $0 \notin T^{-1}$ and $n-1 \in T^{-1}$. Furthermore, we have $0 \in T^{\perp}$ and $n-1 \notin T^{\perp}$. So, $\mathbb{C}_0$ must be the initial cyclotomic coset of $T^{\perp}$. In other words, there must exist an integer $J\geq 1$ such that $T^{\perp}=\mathbb{C}_0\cup \mathbb{C}_1\cup \mathbb{C}_2 \cup \cdots \cup \mathbb{C}_{J-1}$ if $\mathcal{C}_{\delta}$ is a dually-BCH code.

If $\delta_1< \delta \leq n$, it is easy to see that $T^{\perp}=\{0\}$ and $\mathcal{C^{\perp}_{\delta}}$ is a BCH code with respect to $\beta$.
It remains to show that $\mathcal{C^{\perp}_{\delta}}$ is not a dually-BCH code with respect to $\beta$ when $2\leq \delta\ \leq \delta_1$ for $m\geq 2$. To this end, we show that there does not exist an integer $J\geq 1$ such that $T^{\perp}=\mathbb{C}_0\cup \mathbb{C}_1\cup \mathbb{C}_2 \cup \cdots \cup \mathbb{C}_{J-1}$. We have the following four cases for discussion.

\noindent {\bf Case 1:}  $2\leq \delta \leq \frac {q-1}{\lambda}+1$. From Lemma \ref{eq:0918}, we know that $\frac{q^{m}-((\lambda-1)q+1)q^{m-2}-1}{\lambda}\in T^{\perp}$ is a coset leader modulo $n$. It then follows from Lemma \ref{eq:lem919} that
$$I_{max}\,:=max\left\{I(\delta)\,:\,2\leq \delta \leq \frac {q-1}{\lambda}+1\right\}=I(2)=\frac{q^m-1}{\lambda}-q^{m-1}.$$

It is clear that $\frac{q^{m}-((\lambda-1)q+1)q^{m-2}-1}{\lambda} > \frac{q^{m}-1}{\lambda}-q^{m-1}$. Thus, there in no an integer $J\geq 1$ such that $T^{\perp}=\mathbb{C}_0\cup \mathbb{C}_1\cup \mathbb{C}_2 \cup \cdots \cup \mathbb{C}_{J-1}$, i.e., $\mathcal{C^{\perp}_{\delta}}$ is not a BCH code with respect to $\beta.$

\noindent {\bf Case 2:} $\frac {q-1}{\lambda}+1 <\delta\leq \frac{q^{m-1}-1}{\lambda}$. With an analysis similar as in Case 1, from Lemma \ref{eq:lem919} and Lemma \ref{eq:0918}, we obtain that $\mathcal{C^{\perp}_{\delta}}$ is not a BCH code with respect to $\beta.$

\noindent {\bf Case 3:} $\frac {q^{m-1}-1}{\lambda} <\delta\leq \frac{q^{m}-1}{\lambda}-q^{m-1}$. With discussions similar to Case $1$, we can obtain the same results as above.

\noindent {\bf Case 4:} $n-q^{m-1} <\delta\leq \delta_1$. From Lemma \ref{eq:lem919}, we know that
$$I_{max}\,:=max\left\{I(\delta)\,:\, n-q^{m-1} <\delta\leq \delta_1\right\}=1.$$
It is easy to see that $$q^m-\lambda q^{m-1}-1=(q-\lambda-1, q-1,\ldots,q-1)_q$$
is a $q$-cyclotomic coset leader modulo $q^m-1$. Then from Lemma \ref{lem:0913}, $n-q^{m-1}$ is a $q$-cyclotomic coset leader modulo $n$. Hence,
$${\rm CL}(n-\delta_1)\leq n-\delta_1\leq n-(n-q^{m-1})<n-q^{m-1}$$
since $\lambda \, | \, q-1$ and $\lambda\neq q-1$.
From Lemma \ref{proposition-(q^m-1)/(q-1)}, there must exist a $q$-cyclotomic coset leader modulo $n$ in $T^{\perp}$, which is greater than $1$. It then follows that there is not integer $J\geq 1$ such that $T^{\perp}=\mathbb{C}_0\cup \mathbb{C}_1\cup \mathbb{C}_2 \cup \cdots \cup \mathbb{C}_{J-1}$, i.e., $\mathcal{C^{\perp}_{\delta}}$ is not a BCH code with respect to $\beta.$

Combining all the cases, the desired conclusion follows.
\end{proof}
\begin{example}
Let $q=5$, $\lambda=2$ and $m=4$. By Magma, $\delta_1=247$  and $\mathcal{C}_\delta$ is a dually-BCH code if and only if $247 < \delta\leq 312$, where $\delta_1$ is the first largest  $q$-cyclotomic coset leader modulo $312$.  This experimental result coincides with Theorem \ref{thm-04}.
\end{example}

\begin{example}
Let $q=7$, $\lambda=3$ and $m=3$. By Magma, $\delta_1=95$ and $\mathcal{C}_\delta$ is a dually-BCH code if and only if $95 < \delta\leq 114$, where $\delta_1$ is the first largest  $q$-cyclotomic coset leader modulo $114$. This experimental result coincides with Theorem \ref{thm-04}.
\end{example}

\begin{remark}
It seems that determining the values of the first and second largest $q$-cyclotomic coset leaders modulo $\frac{q^m-1}{\lambda}$ is a very hard problem. For some special $\lambda$, the first few largest $q$-cyclotomic coset leaders have been shown in Lemmas \ref{lem:qm1q1}, \ref{lem:qm1q101} and \ref{lem:qm1q102}.
\end{remark}

\section{Conclusion}\label{sec-finals}

 In \cite{GDL21} and \cite{Wang23}, the sufficient and necessary condition for $\mathcal{C}_{\delta}$ being a dually-BCH code and the lower bound on the minimum distance of its dual were developed, where the length of $\mathcal{C}_{\delta}$ is $q^m-1$ or $\frac{q^m-1}{q-1}$.
In this paper, we considered  BCH code $\mathcal{C}_{\delta}$ with length $\frac{q^m-1}{\lambda}$, where $\lambda$ is a positive integer satisfying $\lambda\, |\, q-1$, or $\lambda=q^s-1$ and $s\, |\,m$.
The main contributions of this paper are the following:
\begin{itemize}
\item The sufficient and necessary condition in terms of the designed distance was presented to ensure that $\mathcal{C}_{\delta}$ is a dually-BCH code, where $2\leq \delta\leq n$ (see Theorem \ref{thm-02} and
Theorem~\ref{thm-04}).
\item  Some lower bounds on the minimum distances of the duals of BCH codes were developed (see Theorem \ref{thm:1002} and Theorem \ref{eq:thm9190}). Since we extended the results in \cite{GDL21},  for binary primitive
narrow-sense BCH codes, our bounds on the minimum distances of their duals improve the classical Sidel'nikov bound, and are also better than the Carlitz-Uchiyama bound for a large designed distance $\delta$.
\end{itemize}


\begin{thebibliography}{99}

\bibitem{Aly07} S. A. Aly, A. Klappenecker and P. K. Sarvepalli,  ``On quantum
and classical BCH codes," {\it IEEE Trans. Inf. Theory}, vol. 53, no. 3, pp. 1183--1188, Mar. 2007.

\bibitem{Augot96} D. Augot and F. Levy-dit-Vehel, `` Bounds on the minimum distance of the duals of BCH codes," {\it IEEE Trans. Inf. Theory}, vol. 42, no. 4, pp. 1257--1260, July 1996.

\bibitem{Augot94} D. Augot and N. Sendrier,  ``Idempotents and the BCH bound," {\it IEEE
Trans. Inf. Theory}, vol. 40, no. 1, pp. 204--207, Jan. 1994.

\bibitem{Bose62} R. C. Bose and D. K. Ray-Chaudhuri, ``On a class of error correcting binary group codes," {\it Inf. Control},  vol. 3, pp. 279-290, Mar. 1960 .


\bibitem{Charpin90} P. Charpin, ``On a class of primitive BCH-codes," {\it IEEE Trans. Inf.
Theory}, vol. 36, no. 1, pp. 222--228, Jan. 1990.

\bibitem{Charpin98} P. Charpin, ``Open problems on cyclic codes," in Handbook Coding
Theory, vol. 1, V. S. Pless and W. C. Huffman Eds. Amsterdam,
The Netherlands: Elsevier, 1998, pp. 963--1063, ch. 11.

\bibitem{Ding15} C. Ding, X. Du and Z. Zhou, ``The Bose and minimum distance
of a class of BCH Codes," {\it IEEE Trans. Inf. Theory}, vol. 61, no. 5, pp. 2351--2356, May 2015.

\bibitem{Ding17} C. Ding, C. Fan and Z. Zhou, ``The dimension and minimum distance of two
classes of primitive BCH codes," {\it Finite Fields Appl.}, vol. 45, pp. 237--263, May 2017.

\bibitem{Fan23} M. Fan, C. Li and C. Ding, ``The Hermitian dual codes of several classes of BCH codes," {\it IEEE Trans. Inf. Theory}, vol. 69, no. 7, pp. 4484--4497, Mar. 2023.


\bibitem{GDL21}
B Gong, C. Ding and C. Li, ``The dual codes of several classes of BCH codes," {\it IEEE
Trans. Inf. Theory},  vol. 68, no. 2, pp. 953--964, Mar. 2022.



\bibitem{Gorenstein61} D. C. Gorenstein and N. Zierler, ``A class of error-correcting codes in $p^m$ symbols," {\it J. SIAM}, vol. 9, pp. 207-214, Jun. 1961.



\bibitem{Hocquenghem59} A. Hocquenghem, ``Codes correcteurs d'erreurs," {\it Chiffres}, vol. 2, no. 2, pp. 147-156, 1959.

\bibitem{Lid017} C. Li, C. Ding and S. Li, ``LCD cyclic codes over finite fields," {\it IEEE Trans.
Inf. Theory}, vol. 63, no. 7, pp. 4344--4356, Jul. 2017.

\bibitem{Liu17} H. Liu, C. Ding and C. Li, ``Dimensions of three types of BCH codes over $\gf(q)$,"
{\it Discrete Math.}, vol. 340, no. 8, pp. 1910--1927, Aug. 2017.




\bibitem{Li2017} S. Li, C. Ding, M. Xiong and G. Ge, ``Narrow-sense BCH codes over $\gf(q)$ with length $n=\frac{q^m-1}{q-1}$,"
{\it IEEE Trans. Inf. Theory}, vol. 63, no. 11, pp. 7219--7236, Aug. 2017.

\bibitem{Lid17} S. Li, C. Li, C. Ding and H. Liu, ``Two families of LCD BCH codes,"
{\it IEEE Trans. Inf. Theory}, vol. 63, no. 9, pp. 5699--5717, Sep. 2017.



\bibitem{Li2019} Y. Liu, Y. Li, Q. Fu, L. Lu and Y. Rao, ``Some binary BCH codes with length $n=2^m+1$," {\it Finite Fields Appl.}, vol. 55, pp. 109-133, Jan. 2019.

\bibitem{Ling23} X. Ling, S. Mesnager, Y. Qi and C. Tang, ``A class of narrow-sense BCH codes over $\mathbb{F}_q$ of length
$\frac{q^m-1}{2}$," {\it Des. Codes Cryptogr.}, vol. 88, no. 2, pp. 413--427, Nov. 2020.

\bibitem{MacWilliams77} F. J. MacWilliams and N. J. A. Soane, {\it The theroy error-correcting codes}. Amsterdam. The Netherlands: North-Holland Mathematical Library, 1977.



\bibitem{Wang23} X. Wang, J. Wang, C. Li and Y. Wu. ``Two classes of narrow-sense BCH codes and their duals," {\it IEEE Trans. Inf. Theory}, DOI: 10.1109/TIT.2023.3310193.

\bibitem{Wang19} L. Wang, Z. Sun, and S. Zhu, ``Hermitian dual-containing narrow-sense constacyclic BCH codes and quantum codes," {\it Quantum  Inf. Process.}, vol. 18, no. 323, pp. 1--40, Oct. 2019.



\bibitem{Yue15} D. Yue and Z. Feng, ``Minimum cyclotomic coset representatives
and their applications to BCH codes and Goppa codes," {\it IEEE Trans.
Inf. Theory}, vol. 46, no. 7, pp. 2625--2628, Nov. 2000.

\bibitem{Dianwu96} D. Yue and Z. Hu, ``On the dimension and minimum
distance of BCH codes over $\gf(q)$," {\it J. Electron., China}, vol. 13, no. 3, pp. 216--221, Jul. 1996.


\bibitem{Yan2018} H. Yan, H. Liu, C. Li and S. Yang, ``Parameters of LCD BCH codes with two lengths," {\it Adv. Math. Commun.}, vol. 12, no. 3, pp. 579--594, Dec. 2018.




\bibitem{Zhu19} S. Zhu, Z. Sun and X. Kai,
``A Class of Narrow-Sense BCH Codes," {\it IEEE Trans. Inf. Theory}, vol. 65, no. 8, pp. 4699--4714, Aug. 2019.







\end{thebibliography}
\end{document}